\documentclass[letterpaper, 10 pt, conference]{ieeeconf}  % Comment this line out if you need a4paper

\IEEEoverridecommandlockouts                            
\usepackage{ulem} 

\newcommand{\red}[1]{{#1}}
\newcommand{\blue}[1]{{#1}}
\newcommand{\cyan}[1]{{#1}}

\newcommand{\brown}[1]{{#1}}
\newcommand{\green}[1]{{#1}}

\title{\LARGE \bf
Constructing Dynamic Feedback Linearizable Discretizations.
}

\author{Ashutosh Jindal$^{1}$, Florentina Nicolau$^{2}$, David Mart{\'\i}n Diego$^{3}$ and Ravi Banavar$^{4}$% <-this % stops a space
\thanks{$^{1,4}$
        Systems and Control Engineering, Indian Institute of Technology Bombay, Mumbai, 400076, Maharashtra, India  $^1${\tt\small jindal.ashutosh21@gmail.com}, $^3${\tt\small banavar@iitb.ac.in} }
\thanks{$^{2}$  Quartz EA7393, École Nationale Supérieure de l'Électronique et de ses Applications, Cergy, France
{\tt\small florentina.nicolau@ensea.fr}}
\thanks{$^{3}$ Instituto de Ciencias Matem\'aticas (CSIC-UAM-UC3M-UCM), Calle Nicol\'as Cabrera 13-15, 28049 Madrid, Spain
{\tt\small david.martin@icmat.es}}
}
\usepackage{mymacros}
\usepackage{ragged2e} 
\begin{document}
\maketitle
\thispagestyle{empty}
\pagestyle{plain}

%%%%%%%%%%%%%%%%%%%%%%%%%%%%%%%%%%%%%%%%%%%%%%%%%%%%%%%%%%%%%%%%%%%%%%%%%%%%%%%%
\begin{abstract}
Dynamic feedback linearization-based methods allow us to design control algorithms for a fairly large class of nonlinear systems in continuous time. However, this feature does not extend to their sampled counterparts, i.e., for a given dynamically feedback linearizable continuous time system, its numerical discretization may fail to be so. In this article, we present a way to construct discretization schemes (accurate up to first order) that result in schemes that are feedback linearizable. This result is an extension of our previous work, where 
we had considered only static feedback linearizable systems. The result presented here applies to a fairly general class of nonlinear systems, in particular, our analysis applies to both endogenous and exogenous types of feedback. While the results in this article are presented on a control affine form of nonlinear systems, they can be readily modified to general nonlinear systems. %

\end{abstract}

%%%%%%%%%%%%%%%%%%%%%%%%%%%%%%%%%%%%%%%%%%%%%%%%%%%%%%%%%%%%%%%%%%%%%%%%%%%%%%%%
\pagenumbering{roman}
%\green{I have made changes in green}

%\magenta{I have added my comments in magenta and in todo boxes.}

%\todo[inline]{\justifying Our goal here is not to find or classify systems that are feedback linearizable, nor it is to find out the linearizing output. Instead, we focus on finding a discretization scheme that preserves this linearizability property (as it turns out, with the same output functions) to the discrete-time system.  }

\section{Introduction}
Most engineering systems, from unmanned aerial vehicles to quadrotors, mobile robots, and electric motors are inherently nonlinear in their dynamic behavior. Control techniques for such nonlinear systems are often system-specific. However, linear-time-invariant systems (LTI), that are well studied in the literature, both in continuous-time, as well as discrete time, have many standard control 
algorithms such as  \cyan proportional-integrate-derivative (PID), pole placement and state-feedback, \blue{see e.g., \cite{trentelman2012control,brockett2015finite,rao1997linear} and the references therein}.
% %
% \red{Q: Need to add other references?} -- \magenta{shall do.}
% %
\cyan{Although LTI systems are commonly discussed in theory}, in practice, most physical systems have nonlinearities present in them. Since designing control methods for such nonlinear systems is not that straightforward, one of the techniques used to design control \cyan{laws for a certain class of control systems is based on static} \emph{feedback linearization} (FL). The FL method involves \blue{compensating} the system nonlinearities by
\cyan{changing the coordinate system and applying an invertible static feedback transformation (that can be interpreted as a change of coordinates in the control space, depending on the state) such that the transformed} system dynamics take the form of
an LTI system. This allows us to lift linear control methods and helps in synthesizing control \cyan{laws} for \emph{feedback linearizable} nonlinear systems.

\cyan{Many systems  of engineering interest}
\green{such as the induction motor \cite{dynind}, and the wheeled robot \cite{DELUCA2000687} cannot be linearized by static feedback linearization alone. However, such systems \cyan{can be rendered static feedback linearizable by the application of a dynamic compensator.}
%extending the state space by augmenting additional states.
A bunch of such \cyan{mechanical} systems are cataloged in \cite{murray1995differential}.In contrast to static feedback, for a \textit{dynamic feedback},  or \textit{dynamic compensation}, \cyan{the original controls are not computed from the new ones
by simply static functions, but through a dynamic system which has a
certain state.} Hence dynamic feedback involves extending system states such that the augmented system becomes \cyan{static} feedback linearizable, thus extending the notion of \cyan{static} feedback linearizability to a larger class of systems called the \textit{dynamic feedback linearizable systems}.}
 Dynamic feedback linearizability is closely related to the notion of differential flatness
  \cite{fliess1995flatness, fliess1999lie}. Indeed, flat systems are linearizable via endogenous dynamic feedback, see \cite{fliess1995flatness, levine2009analysis, martin1992contribution,martin2001flat}.
  The analysis present in this paper is general enough and does not require strong assumptions on the underlying dynamics and types of feedback considered. In particular, we do not restrict our analysis to the class of endogenous dynamic feedback and the \cyan{presented results are actually valid for both endogenous and exogenous dynamic feedback.} 

%  \red{@Ashutosh, is the above statement in red correct? -- \magenta{yes, ma'am this is correct. We do not use the endogenous property in Theorem~\ref{main_result}.}} 
  
The application of dynamic feedback linearizability (especially of differential flatness) in addressing engineering challenges has experienced notable growth in recent years. Control design based on dynamic feedback linearizability was applied for important problems in control theory such as motion planning, constructive controllability, or trajectory tracking,  as shown by numerous works (see, e.g., {\cite{della2020model, kolar2017time, mokhtari2004dynamic, messaoudi2023flatness, sun2022comparative, tang2011differential}}).

% \red{@Ashutosh: we can remove some of the above references if not enough space.I removed the subsection title and added the last paragraph to link both parts of the introduction}

\color{black}

%%%% 
% \green{Several systems such as the induction motor \cite{dynind}, and the wheeled robot \cite{DELUCA2000687} cannot be linearized by static feedback linearization alone. However, such systems are rendered feedback linearizable by extending the state space by augmenting additional states. A bunch of such systems are cataloged in \cite{murray1995differential}. This process is called \textit{dynamic compensation} and it involves extending system states such that the augmented system becomes feedback linearizable by endogenous feedback (see \cite{martin1992contribution,martin2001flat} ). Such compensation allows us to extend the notion of feedback linearizability to a larger class of systems called the \textit{dynamic feedback linearizable systems}.}
%%%%

% \red{Q: Why do we need "partial" FL?} -- \magenta{It is just to accommodate the systems for which the relative degree may not be $n$, and to accommodate input-output linearization}

%\subsection{Effects of sampling on feedback linearization}

While most systems evolve in continuous-time, control design and implementation are invariably done in the digital domain. To implement such continuous-time control laws digitally, the dynamics must be discretized. This is often done synchronously by using sample and hold techniques. The sensor values are read at regular intervals (instead of a continuous measurement) and a piecewise constant control (with the control value held constant between two successive intervals) is applied to the actuators. Furthermore, to study dynamical systems digitally, one needs to evaluate the evolution of the system over two successive intervals. Often, such an evolution is unavailable in a closed-form analytical expression and is to be approximated numerically. Some of the commonly used numerical integrators are Eulerian \brown schemes, Runge-Kutta-based methods, etc.

Sample and hold restricts the choice of available controls to the set of piecewise controls, and this in general does not preserve the feedback linearizability property \blue{of the original continuous-time control system}. In other words, a given \brown{continuous-time system that is (static or dynamic)} feedback linearizable, \green{it may} not remain feedback linearizable after discretization. It has been established in \cite{grizzle1988feedback} that feedback linearizability is not always preserved under discretization and is also dependent on the choice of discretization.

In \cite{retr_disc_map}, we show that for a continuous-time system that can be linearized by static feedback, using discretization maps one can construct discretization schemes that are accurate up to first-order and preserve feedback linearizability. Such a discretization scheme allows us to leverage the feedback linearization-based methods to design control for the nonlinear systems. In this article we extend the results of \cite{retr_disc_map} to systems linearizable by dynamic feedback. For a given dynamic feedback linearizable system, we present a systematic way to construct discretization schemes on the extended system such that the resulting discrete-time system is feedback linearizable in the discrete-time.

\section{Static and Dynamic Feedback Linearization}
Let $\stateset\subset\R[n]$ and $\controlset\subset\R[m]$ be nonempty and open
\cyan{(more
generally, $n$-dimensional  and an $m$-dimensional manifolds, respectively).}
Let $\stateset\ni x\mapsto f(x)\in\R[n]$, and $\stateset\ni x\mapsto g_i(x)\in\R[n]$ for all $1\leq i\leq m$ be sufficiently smooth. Then a control-affine continuous-time system evolving on $\stateset\times\controlset$, with $m$ \cyan{inputs} is given by a differential equation
\begin{equation}
\label{ctssys}
\begin{split}
\dot{x}(t) &= f(x(t))+\sum_{i=1}^m g_i(x(t))u_i(t)\\
&= f(x) + g(x)u
\end{split}
\end{equation}
where \blue{$x \in\stateset$, $[u_1\ldots u_m]^\top\eqqcolon u\in\controlset$} denote the system state and control input, respectively. \brown{We assume that the vector fields $g_i$ are everywhere independent.}
\cyan{Note that we consider the control-affine case since for most engineering applications, the dynamics of the plant can be modeled with control-affine systems. However, all results are valid for general control systems (nonlinear with respect to the control) of the form $\dot x = F(x,u)$.}
\cyan{Most notions recalled in this section  are commonly covered in classical nonlinear control theory textbooks (for a comprehensive overview, see, e.g., the recent textbook \cite{lee2022linearization}
that specifically focuses on linearization of nonlinear control system, and the references therein).}

% \todo[inline]{
% \red{@Ashutosh, why do you cite \cite{lee2022linearization} here? It is strange to give a citation for the definition of 2022 and to say just after the definition that the problem has been solved in the '80s.
% \\
% I removed the reference of the book from the definitions and added the following:
% -- \magenta{Sure, the reason for citing \cite{lee2022linearization} was that it compiles the major results on feedback linearization of continuous and discrete time systems in a single reference.}
% }
% }

\green{
\begin{defn}[Static Feedback Linearization]
\label{def: static FL}
For some $x_0\in\stateset$, let $\open{x_0}\ni x_0$ be an open \blue{neighborhood of  $x_0$} in~$\stateset$. System~\eqref{ctssys} is \textit{static feedback linearizable} on $\open{x_0}$, if there exists a state transformation
\begin{equation}
    \open{x_0}\ni x\mapsto \phi(x)\eqqcolon z\in\open{z_0},
\end{equation}
where $z_0\coloneqq \phi(x_0)$ and $\phi$ is a diffeomorphism onto its image $\phi(\open{x_0})\eqqcolon \open{z_0}$ and a static feedback
\begin{multline}
\label{stat_feed}
\open{x_0}\times\R[m]\ni (x,v)\mapsto\\ \alpha(x)+\beta(x)v\eqqcolon u\in \controlset\subset\R[m]
\end{multline}
with $\beta(x)\in\R[m\times m]$ nonsingular for all $x\in\open{x_0}$, such that~\eqref{ctssys} is \cyan{transformed into} the following linear time invariant system:
\begin{equation}
    \dot{z}(t) = Az(t)+Bv(t),
\end{equation}
where  $A\in\R[n\times n]$, and  $B\in\R[n\times m]$  are constant matrices such that $ A \phi(x) = \D{\phi}(x)\cdot\big(f(x)+g(x)\alpha(x)\big)$,  $B = \D{\phi}(x)\cdot g(x)\beta(x)$.
%where $A\in\R[n\times n]$, $B\in\R[n\times m]$ are fixed matrices.
\end{defn}
}
%
%\red{@Ashutosh: This definition becomes complicated when we want to take into account the subsets $\open{v_0}$ and $\open{u_0}$, which are not actually needed, since as I explained below, static FL for control-affine systems is global with respect to $u$ (I agree that the dynamic feedback introduces singularities in the control space, but for Def 2.1 we do not need the local character with respect to $u$).
%\\
%Also this is not consistent with definition  2.2 where you take $\R[m]$ for $\mu$ in (2.13).
%}
\begin{rmk}
For control affine systems such as \eqref{ctssys}, the notion of feedback linearization is global with respect to control, i.e, for all $v\in\R[m]$ (such that $u\in\controlset$) transformation $v\mapsto\alpha(x)+\beta(x)v\eqqcolon u$ is well-defined. For general nonlinear system $\dot{x} = F(x,u)$, the general nonlinear feedback $u = \gamma(x,v)$ is defined locally with respect to both control and state.
\end{rmk}
For necessary and sufficient conditions under which a given system is (static) feedback linearizable see
  \cyan{\cite{l1980linearization, hunt1981linear} (see also,
 \cite{isidori1982feedback, book:isidori1995, brockett1978feedback} for related results).}

\cyan{The above class of transformations can be enlarged by considering dynamic feedback, as the following example shows.}

 \begin{example}
Consider the following example 
\label{ex_1}
\begin{equation}
\label{unicycle}
\pmat{\dot x_1\\\dot x_2\\\dot x_3\\\dot x_4} = 
\pmat{x_2+2x_2x_3\\x_3\\0\\0}+\pmat{0\\0\\1\\0}u_1+\pmat{2x_2x_4\\x_4\\0\\1+x_3}u_2,
\end{equation}
\brown{around any $x_0\in \R[4]$ such that $(x_{20}, x_{30}, x_{40})\neq (0, -1, 0)$,} 
where $x\coloneqq(x_1,x_2,x_3,x_4)\in\R[4],$ and $u\coloneqq(u_1,u_2)\in\R[2]$ denotes the system state and control input respectively. \cyan{It can be easily shown that~\eqref{unicycle} does not satisfy the
necessary and sufficient conditions for static  feedback linearizability \cite{l1980linearization, hunt1981linear}, thus it is not static  feedback linearizable.}
However, if one instead considers the following \cyan{dynamic} precompensator
\begin{equation}
\begin{split}
u_1 &= \extendedcontrol_1\\
u_2 &= w\\
\dot{w}& = \extendedcontrol_2\\
\end{split}
\end{equation}
and the extended system dynamics given by
\begin{equation}
\begin{split}
\label{unicycle_extended}
\pmat{\dot x_1\\\dot x_2\\\dot x_3\\\dot x_4\\\dot{w}} &= 
\pmat{x_2+2x_2(x_3+x_4w)\\x_3+x_4w\\0\\(1+x_3)w\\0}+\pmat{0\\0\\1\\0\\0}\extendedcontrol_1+\pmat{0\\0\\0\\0\\1}\extendedcontrol_2,
\end{split}
\end{equation}
then one can see that under the coordinate transformation \brown{(locally invertible around any $(x_0, w_0)$ such that $1+x_{30} - w_0 x_{40} \neq 0$)}
\begin{multline*}
(z_1,z_2,z_3,z_4,z_5)\coloneqq \extdiff(x_1,x_2,x_3,x_4,w)\\
\coloneqq (x_1-x_2^2,x_2,x_3+x_4w,x_4,(1+x_3)w)
\end{multline*}
and \brown{the invertible} static feedback
\begin{equation*}
\pmat{\extendedcontrol_1\\\extendedcontrol_2}\coloneqq \pmat{1&x_4\\ w& 1+x_3}^{-1}\pmat{v_1-(1+x_3)w^2\\v_2},
\end{equation*}
where $ v \coloneqq (v_1,v_2)$ is the modified control input,~\eqref{unicycle_extended} is equivalent to the following LTI system
\begin{equation}
\label{unicycle_lin}
\begin{split}
\dot{z_1}&=z_2\\
\dot{z_2}&=z_3\\
\dot{z_3}&=v_1\\
\dot{z_4}&=z_5\\
\dot{z_5}&=v_2
\end{split}
\end{equation}
and therefore, \brown{the original system~\eqref{unicycle} is dynamic feedback linearizable around $(x_0, u_0)$ such that $1+x_{30} - u_{20} x_{40} \neq 0$)}.
\end{example}

 %\red{@Ashutosh: 1) Again why this particular reference for the next two definitions? Cite earlier works.
% \\
 %2) Why do you state ``\eqref{dyn_comp} is  \red{invertible}''?  which is the meaning of ``invertible'' here?
 %is it $ \extendedcontrol = \psi(x, w, u, \dot u, \ddot u, \ldots)$?
% what are exactly the conditions that you require/need for the Dynamic Compensator ?
 %\\
 %3) You do not need to define the Dynamic Compensator locally.
 %}

\begin{defn}[Dynamic Compensator]
For some given $w_0\ni\R[q]$ (set to zero without loss of generality), let $\open{w_0}\ni w_0$ be \blue{an open neighborhood of $w_0$} in $\R[q]$, and $\extendedcontrol\in\R[m]$, let $\alpha,\beta, \gamma,\delta,$ be smooth maps on $\open{x_0}\times\open{w_0}$ mapping into suitable codomains, then a dynamic compensator is given by
\begin{equation}
\label{dyn_comp}
\begin{split}
   \dot w &= \gamma(x,w)+\delta(x,w)\extendedcontrol,\\
   u &= \alpha(x,w) +\beta(x,w)\extendedcontrol.  
\end{split}    
\end{equation}
where $\extendedcontrol\in\R[m]$ and $\alpha,\beta,\gamma,\delta$ are sufficiently smooth.
\end{defn}
\begin{rmk}
Note, apart from regularity assumptions regarding smoothness, we do not assume any additional requirements on $\alpha,\beta,\gamma,\delta$. Further, the affine structure of the compensator is only guaranteed for control-affine systems of type~\eqref{ctssys}. For a general nonlinear system, the compensator is defined by general nonlinear maps of the type $\dot w = \Lambda(x,w, \mu)$ and $\extendedcontrol = \Gamma(x,w,v)$.     
\end{rmk}
\begin{defn}[Dynamic Feedback Linearization]
For some given $x_0\in\stateset$ and $u_0\in\controlset$, let $\open{x_0}$ and $\open{u_0}$ be open. System~\eqref{ctssys} is said to be linearizable by dynamic feedback of type~\eqref{dyn_comp} \brown{if} the extended system  
\begin{equation}
\label{ext_sys}
\begin{split}
\dot{x} &= f(x)+g(x)(\alpha(x,w)+\beta(x,w)\extendedcontrol)\\
\dot{w} &= \gamma(x,w) +\delta(x,w)\extendedcontrol
\end{split}    
\end{equation}
which when written compactly as 
\begin{equation}
\label{ext_compact}
 \dot{\extendedstate} = F(\extendedstate) + G(\extendedstate)\extendedcontrol   
\end{equation}
where $\extendedstate\coloneqq(x,w)\in\open{x_0}\times\open{w_0}$ and $\extendedcontrol\in\R[m]$ are the states and control of the extended system respectively, 
is static feedback linearizable on $\open{x_0}\times\open{w_0}$, i.e., there exists a 
\begin{multline}
\label{phi_1}
  \R[n]\times\R[q]\supset  \open{x_0}\times\open{w_0}\ni(x,w)\mapsto\\ \extdiff(x,w)\eqqcolon \linstate\in\R[n]\times\R[q]
\end{multline}
and a feedback  
\begin{multline}
\label{psi_1}
\open{x_0}\times\open{w_0}\times\R[m]\supset (x,w,v)\mapsto \\ \extctra(x,w)+\extctrb(x,w)v \eqqcolon \extendedcontrol\in\R[m]
\end{multline}
such that~\eqref{ext_sys} is %\sout{equivalent} 
\brown{transformed into} a linear dynamical system 
\begin{equation}
\label{dyn_linear}    
\dot{\linstate}(t) = A\linstate(t)+Bv(t)
\end{equation}
where $\linstate(t) = \extdiff(x(t),w(t))$ for all $t$, and $A\in\R[(n+q)\times (n+q)]$ and $B\in\R[(n+q)\times m]$ are such that $A\extdiff(x)= \D F (\extendedstate)\cdot\big(F(\extendedstate)+G(\extendedstate)\extctra(\extendedstate)\big)$ and $B = \D {G}(\extendedstate)\cdot\extctrb(\extendedstate)$.
\end{defn}

\section{Constructing Feedback Linearizable Discretization}
While both static and dynamic feedback linearization methods 
\brown{transform a given 
nonlinear system into a linear one,} 
%render a given nonlinear system (state) input-output linear, 
these properties are not preserved under-sampling as demonstrated by Grizzle in \cite{grizzle1988feedback} (for information on discrete time feedback linearization see \cite{grizzle1986feedback,grizzle1988feedback,lee1987linearization} and references therein). Moreover, the choice of discretization plays a key role in the feedback linearizability of the resulting discretized system. Consider the explicit Euler discretization of the dynamically compensated system~\eqref{unicycle_extended} from Example~\ref{ex_1}: 
\begin{equation}
\label{uni_dis_eul}
\pmat{x^1_{k+1}\\x^2_{k+1}\\x^3_{k+1}\\x^4_{k+1}\\w_{k+1}} = \pmat{x^1_{k}\\x^2_{k}\\x^3_{k}\\x^4_{k}\\w_{k}}+h\pmat{x^2_k+2x^2_k(x^3_k+x^4_kw_k)\\x^k_3+x^4_kw_k\\\extendedcontrol^1_k\\(1+x^3_k)w_k\\\extendedcontrol^2_k}
\end{equation}
where for all $k\in\N$, $(x^1_k,x^2_k,x^3_k,x^4_k,w_k)\eqqcolon\extendedstate_k\approx \extendedstate(t_k)$ is the approximated trajectory of~\eqref{unicycle_extended}, $(\extendedcontrol_k^1,\extendedcontrol_k^2)\eqqcolon\extendedcontrol_k \in\R[2]$ is the piecewise constant control input applied over interval $[t_k,t_{k+1}[$, and $h$ is the sampling period with $t_{k+1} = t_{k}+h$. Using \cite[Theorem~3.1]{grizzle1986feedback}, one can check that~\eqref{uni_dis_eul} is not feedback linearizable (see appendix for calculations). 
%\\
%\todo[inline]{\eqref{uni_dis_eul} is feedback linearizable} 
However, if one chooses an alternate discretization scheme (\brown{whose construction is detailed in Section~\ref{sec:illustrating_ex}}, see~\eqref{ex1_disc_nonlin}), the resulting discrete-time system is then feedback linearizable. 

In \cite{retr_disc_map}, we have demonstrated that for a given (static) feedback linearizable continuous-time system, using discretization maps it is possible to construct discretization schemes, that preserve feedback linearizability. As the main contribution of this article, we now \brown{show} how these earlier results can be extended to dynamic feedback linearizable systems. Before stating the main result, we provide a rapid refresher on the retraction and discretization maps. For more information on these maps one may look into \cite{21MBLDMdD} and references therein.  

\subsection{Retraction and Discretization maps}
%\red{TODO: here explain the goal of this section and relate it with the considered problem. Verify the first paragraph and keep only what is needed. 
%\\
%Q: what is $k$ below? 
%You define only Discretization maps. Do you need also Retraction maps? (if not, change the title of the subsection and adapt the following.)
%}
Euclidean methods such as the Eulerian scheme give satisfactory performance for systems evolving on linear vector spaces, however for systems evolving on nonlinear manifolds such as $\SO{3}$, such schemes do not guarantee system states to remain on the manifold for all time instants. This results in erroneous performance as the trajectory of the numerically discretized system no longer satisfies the geometric constraints\footnote{\green{Here geometric constraints imply the constraints describing the system manifold, \brown{that we denote in this section by} $M$.}} of the continuous-time system.

Retraction and discretization maps are a class of maps that utilize the geometric properties of the manifold to construct discretizations such that the system states remain on the manifold \red{for all $k\in \N$, where $k$ is the iterating index of the discretized trajectory}.

%\red{TODO: I think that we should explain and formalize in this subsection what we mean by $k$, $x_k$, $u_k$, etc (something similar to what is done in footenote 2). We only briefly  presented it for the example, just after (3.1). -- \magenta{sure}}

Let $M$ be an $n$-dimensional manifold (not necessarily associated with system \eqref{ctssys}) and $\T M$ be the associated tangent bundle. Let $\T M\ni(x,\dot{x})\mapsto  \pi_M(x,\dot{x}) = x$ be the canonical projection on $M$ and $0_x\in\T_xM$ be the zero vector in $\T_x M$.  

%\red{TODO: standardize notations. The notations of this section are not consistent with those of Section II. For example, in Section II, $v$ is a control, the state manifold is denoted by $X$, $U$ is associated to the control, etc.
%\\
%Suggestion: 
%\\
%- keep $M$ here but say that it is not necessarily associated with the original system;  
%\\
%- either use $\dot x$ instead of $v$ or denote by $\nu$ the control of the linear system in Section II and keep $v$ here -- \magenta{sure, I will replace the $v$ with $\dot{x}$, I am sorry for this, I copied the definition from the previous paper, and forgot to change the notation}
%\\
%- replace $U$  in Def 3.2 by $\mathcal O$
%}

\begin{defn}[Retraction Maps {\cite{AbMaSeBookRetraction}}]
Consider a smooth map $\retmap\colon \T M\lra M$  and $\retmap_x \eqqcolon \retmap\restrto{\T_x{M}}$ be its restriction onto $\T_{x}M$ then $\retmap$ is a retraction if for all $x\in M$, 
\begin{enumerate}
    \item $\retmap_x(0_x)=x$,  and,
    \item $\T_{0_x}\retmap_x$ is the identity map on $\T_xM$.  
\end{enumerate}
\end{defn}
Retraction maps can be generalized to define the discretization maps as follows:
\begin{defn}[Discretization Maps {\cite{21MBLDMdD}}]
Let $\mathcal{O}\subset \T M$ be an open neighborhood of the zero section of the tangent bundle $TM$. $\mathcal{O}\ni(x,\dot{x})\mapsto \discmap(x,\dot{x})\coloneqq (\discmap^1(x,\dot{x}),\discmap^2(x,\dot{x}))\in M\times M$ is a discretization map if, for any $x\in M$, it satisfies 
\begin{enumerate}
    \item $(x,0_x)\mapsto \discmap(x,0_x) = (x,x)$, and
    \item $\T_{(x,0_x)}\discmap^2-\T_{(x,0_x)}\discmap^1 \colon \T_{(x,0_x)}\T_xM\simeq \T_xM \lra T_xM$ is the identity map on $T_xM$,
\end{enumerate}
where $\T_{(x,0_x)}\discmap^i$ is the tangent map of $\discmap^i$, $i\in\{1,2\}$ at $(x,0_x)\in \T M$, and $\T_{(x,0_x)}\T_xM$ is canonically identified with $\T_xM$. 
\end{defn}
\begin{rmk}
One natural way to construct discretization maps from a retraction map is as follows: Let $\retmap$ be a retraction map on $\T M$, then $\T M\ni(x,\dot{x})\mapsto \discmap(x,\dot{x})\coloneqq (x, \retmap(x,\dot{x}))\in\T M$ is a discretization map on $M$.

Another key feature of the discretization (and retraction) maps is that discretization maps are preserved {under diffeomorphisms} between manifolds. This allows us to lift discretization maps between manifolds. 
\end{rmk}

\begin{prop}[Lift of discretization maps \cite{retr_disc_map}]
\label{r_lift}
Consider two $n$-dimensional manifolds $M$ and $N$. Suppose $M\ni x\mapsto \phi(x)\eqqcolon y\in N$ is a diffeomorphism. Then for a given discretization map $\discmap$ on $M$, $\discmap_\phi := (\phi\times\phi)\circ \discmap\circ T\phi^{-1}$ is a  discretization map on $N$ (see Figure~\ref{fig:my_label}).    
\end{prop}

Further for each given discretization map, {and (controlled) vector field} one can construct a first-order\footnote{\green{A numerical approximation $x_{k+1} = F_h(x_k,u_k)$ for a continuous-time system is called of order $r$ if there exist $K>0$ and some $h_0>0$ such that for all $0<h<h_0$, and  $\norm{x(t_{k+1})-F_h(x(t_k),u_k))}/h\leq Kh^r$, where $t_{k+1}= t_k+h$ and $t\mapsto x(t)$ is a solution of the continuous-time system \cite{blanes2017concise}.}} discretization scheme. 
% \todo[inline]{
% \textcolor{blue}{This is a matter of notation. I prefer the retraction from the continuous linear system to the discrete linear system to be called $R$ and the induced one to be $R_{\phi}$. But it is not important. -- \magenta{setting $\discmap_\phi$ as the discretization on the linearized coordinates as it avoids the subscript $\discmap_{\phi^{-1}}$ on $M$. In the sense $\discmap_{\phi}$ is pulled back by $\phi^{-1}$ to get $R$.}}
% }

%\red{Q: where do you need the notation %$\mathfrak{X}(M)$?\\
%Should we say something about $k$ below? at least "for all $k\in \N$"?
%\\
%Here is $F$ any or affine wrt u?  Maybe say that in this section we consider general control systems? 
%}

\begin{prop}[Discretization of vector fields \cite{retr_disc_map}]
\label{first_desc}
For each \blue{fixed $u\in \mathbb{R}^m$}, let $M\ni x\mapsto \F_u(x)\coloneqq(x,F(x,u))\in\T M$ be a controlled vector field on $M$. Then for a given step size $h>0$ and for each $k\in\N$ %and a fixed time-%discretization map %$t\mapsto(t-\alpha h, t+(1-%\alpha)h)$, %$\alpha\in[0,1]$,
\begin{equation*}
\discmap^{-1}(x_k, x_{k+1})=h\F_{u_k}\underbrace{( \pi_M( \discmap^{-1}(x_k, x_{k+1})) )}_{\in M}
\end{equation*}
%where $x_k = x(t-\alpha h)$ and $x_{k+1} = x(t+(1-\alpha)h)$,
is  a first-order discretization of $\dot{x} = F(x,u)$ with $x_k\approx x(t_k)$, where the sequence $\{t_{k}\mid k\in\N, t_{k+1}=t_k+h\}$ denotes the time instances at which states are sample and $x(t_k)$ is the exact trajectory of $\dot{x}= F(x,u)$\footnote{Here we have considered the general nonlinear form as the assertions made here hold true for any nonlinear system.}.  
%and second-order if $\discmap$ is symmetric.      
\end{prop}

\begin{figure}
    \centering
    \begin{tikzpicture}
  \matrix (m) [matrix of math nodes,row sep=3em,column sep=4em,minimum width=2em]
  {
     \T M\ni (x,\dot{x}) & (y,\dot{y})\in\T N \\
     M\times M\ni\discmap(x,\dot x) & \discmap_\phi(y,\dot y)\in N\times N\\};
  \path[-stealth]
    (m-1-1) edge node [left] {$\discmap$} (m-2-1)
            edge node [above] {$\T\phi$} (m-1-2)
    (m-2-1.east|-m-2-2) edge node [below] {$\phi\times\phi$}
            (m-2-2)
    (m-1-2) edge node [right] {$\discmap_{\phi}$} (m-2-2);
\end{tikzpicture}
    \caption{$\discmap$ and $\discmap_\phi$ commute as shown above}
    \label{fig:my_label}
\end{figure}

%\red{TODO: Verify the above equation of $R_{\Phi}^{-1}$, it is not consistent with that of PRop 3.2} -- \magenta{corrected}

%\red{TODO: give more comments, and intuitions about the results presented in Section III A. Maybe a simple example.} -- \magenta{The results in Section III A are taken from \cite{retr_disc_map} and \cite{21MBLDMdD}.}

\blue{\subsection{Main result}}
Proposition~\ref{first_desc} along with Proposition~\ref{r_lift} allows us to construct discretizations that are feedback linearizable as we now demonstrate. The idea is first to construct a discretization scheme for the \brown{linear} system~\eqref{dyn_linear} and then lift it using the diffeomorphism $\extdiff$ to construct discretizations for  \brown{the extended system}~\eqref{ext_sys}. 

%\red{To define the relative degree, we need an output. What is the output here? IO linearization is a ''weaker`` result than feedback linearization. For what kind of linearization, does the proposed method work?} -- \magenta{The method will work fine for both input-output linearization, as well as static and dynamic feedback linearization. -- for I/O linearizable system, corresponding discrete-time system will also be I/O linear, similarly for dynamic feedback linearizable system. -- I have made changes in the main result.}

%\red{Retraction map not defined} -- \magenta{discretization maps are a generalization of retraction maps\cite{21MBLDMdD}, I will make the edits in the document}

Denoting $M \coloneqq \open{x_0}\times\open{w_0}\subset\R[n]\times\R[q]$ and $N\coloneqq \extdiff(\open{x_0}\times\open{w_0})$, 
\brown{where $\extdiff$ is defined by~\eqref{phi_1}, the map} 
$\extdiff\colon M\lra N$ is then a diffeomorphism. Let $$\T N\ni(\linstate,\dot\linstate)\mapsto \discmap_\extdiff(\linstate,\dot\linstate)\in N\times N$$  be a \green{discretization map} such that the induced discretization of~\eqref{dyn_linear}
\begin{equation*}
\discmap_{\extdiff}^{-1}(\linstate_{k},\linstate_{k+1}) = h\Flin_{v_k}\Big(\pi_N(\discmap^{-1}_\extdiff((\linstate_{k},\linstate_{k+1})))\Big)
\end{equation*}
%\red{Too many "(" and ")" above}
where for each \brown{fixed} $v$, $N\ni z\mapsto \Flin_{v}(\linstate) = (z,A\linstate+Bv)\in\T N$ is of the form
\begin{equation}
\label{dis_lin}
\begin{split}
    \linstate_{k+1} &= A_h\linstate_k+B_hv_k
\end{split}    
\end{equation}
where $A_h\in\R[(n+q)\times(n+q)]$, $B_h\in\R[(n+q)\times m]$ are fixed matrices.

%\red{Q: again why ''input-output linearizable``?} -- \magenta{corrected}
\begin{thm}
\label{main_result}
Consider a \brown{dynamic feedback linearizable} continuous-time system given by~\eqref{ctssys}, such that its dynamic extension~\eqref{ext_sys} \brown{can be transformed into the linear system}~\eqref{dyn_linear}. Let $\extdiff$ be as in~\eqref{phi_1} \brown{and} $\discmap_\extdiff$ be a discretization map on $N$ resulting in a discretization scheme~\eqref{first_desc}. \brown Then there exists a discretization map on $M$ given by
\begin{equation}
\label{dis_nonlin}
    \discmap \coloneqq (\extdiff^{-1}\times\extdiff^{-1})\circ \discmap_\extdiff\circ \T\extdiff
\end{equation}
inducing  the following first-order discretization scheme on~\eqref{ext_sys} 
\begin{equation}
\label{retr_desc_lin1}
\discmap^{-1}(\extendedstate_{k},\extendedstate_{k+1})= h\F_{\extendedcontrol_{k}}\Big(\pi_M(\discmap^{-1}(\extendedstate_{k},\extendedstate_{k+1}))\Big), 
\end{equation}
where for each \brown{fixed} $\extendedcontrol$,  $M\ni\extendedstate\mapsto \F_{\extendedcontrol}\coloneqq (\extendedstate, F(\extendedstate)+G(\extendedstate)\extendedcontrol)\in\T M$, and $h\F_{\extendedcontrol}(\extendedstate) \coloneqq (\extendedstate,h(F(\extendedstate)+G(\extendedstate)\extendedcontrol))$,
such that it is (static) feedback linearizable in the discrete-time sense. Moreover, the discrete linearizing feedback is given by 
\begin{equation}
\label{lin_feed_disc}
\extendedcontrol_k =  \extctra(\extendedstate_k)+\extctrb(\extendedstate_k)v_k.
\end{equation}
\end{thm}
\begin{proof}
Since $\discmap_\extdiff$ is a discretization map on $N$, using Proposition~\ref{r_lift}, $\discmap$ is a discretization map on $M$. Further, using \red{Proposition~\ref{first_desc}}, we have

%\red{Label problem : \ref{retr_desc_lin1} is not a proposition}
\begin{equation}
    \discmap^{-1}(\extendedstate_{k},\extendedstate_{k+1}) = h\F_{\extendedcontrol_k}\big(\pi_M(\discmap^{-1}(\extendedstate_k,\extendedstate_{k+1}))\big).
\end{equation}
From this, \brown{it follows}
\begin{equation*}
(\extendedstate_k,\extendedstate_{k+1})= \discmap\Big(h\F_{\extendedcontrol_k}\big(\pi_M(\discmap^{-1}(\extendedstate_k,\extendedstate_{k+1}))\big)\Big)    
\end{equation*}
which implies
\begin{multline*}
(\extdiff\times\extdiff)(\extendedstate_k,\extendedstate_{k+1})\\= (\extdiff\times\extdiff)\left(\discmap\Big(h\F_{\extendedcontrol_k}\big(\pi_M(\discmap^{-1}(\extendedstate_k,\extendedstate_{k+1}))\big)\Big)\right).
\end{multline*}
Substituting \red{$\discmap$ by its expression~\eqref{dis_nonlin}} in the above equation, we have 
%\red{@A: is this what you wanted to say? verify the label}
\begin{equation*}
\begin{split}
(\linstate_{k},\linstate_{k+1}) &=  \discmap_{\extdiff}\circ\T\extdiff\circ h\F_{\extendedcontrol_k}\big(\pi_M(\T\extdiff^{-1}\circ \discmap_{\extdiff}^{-1}(\linstate_k,\linstate_{k+1}))\big).
\end{split} 
\end{equation*}
Substituting $\extendedcontrol_k = \extctra(\extendedstate_k)+\extctrb(\brown{\extendedstate_k})v_k$, we \brown{get} 
\begin{multline*}
\T\extdiff\circ\F_{\mu_k}\big(\pi_M(\T\extdiff^{-1}\circ \discmap_{\extdiff}^{-1}(\linstate_k,\linstate_{k+1}))\big)= \\
\Flin_{v_k}\big(\pi_N(\discmap_{\extdiff}^{-1}(z_k,z_{k+1}))\big)
\end{multline*}
where $\Flin_{v_k}(\linstate_k) = (z_k,A\linstate_k+Bv_k)\in\T N$ and $(\linstate,\dot\linstate)\mapsto \pi_N(\linstate,\dot\linstate)=\linstate$ is the canonical projection on to $N$. Thus, we have,
\begin{equation*}
    \discmap_{\extdiff}^{-1}(\linstate_k,\linstate_{k+1}) = h\Flin_{v_k}\Big(\pi_{N}(\discmap_{\extdiff}^{-1}(\linstate_k,\linstate_{k+1}))\Big).
\end{equation*}
Since $\discmap_\extdiff$ induces a discretization that preserves linearity, we \brown{finally obtain}
\begin{equation*}
\begin{split}
    \linstate_{k+1} = A_h\linstate_k+B_hv_k,
\end{split}
\end{equation*}
thereby concluding the proof.
\end{proof}
%\blue{I think that it is necessary to discuss the relation between $v_k$ and $\mu_k$.}
% \todo[inline]{We present Corollary~\ref{cor_part_lin} only to show that it is possible to construct discretization that can preserve the partial linearized subsystem. This makes the result more generic and applicable to a larger class of systems (including input-output linearization) and can be removed if space constraint}
%\red{Q: what is the meaning of ''delayed'' here? this may be confusing} --\magenta{this has been corrected}
%\begin{rmk}
%While implementing~\eqref{retr_desc_lin1}, the extended control input $\extendedcontrol_k$ is used to compute the control input $u_k$ from~\eqref{ext_sys}. The control input $u(t)$ is then held constant at $u_k$ over the interval $t\in[t_{k},t_{k+1}[$ and is applied to~\eqref{ctssys}.
%\end{rmk}
The discrete-time evolution is obtained by solving~\eqref{retr_desc_lin1} implicitly for $\extendedstate_{k+1}$, for each given $\extendedstate_k$ and $\extendedcontrol_k$, $k\in\N$. Writing it in closed form\footnote{Existence of $\bar F_h$ in a local neighborhood is guaranteed using the Implicit Function Theorem \cite[Theorem~4.B]{zeidler1993vol}, however, the actual closed-form expression of $\bar F_h$ is not available in general.} 
\begin{equation*}
    \extendedstate_{k+1} = \bar F_h\blue({\extendedstate_k,\extendedcontrol_k}\blue)
\end{equation*}

%\red{Q: Is this affine wrt $\mu_k$, then we should have: -- \magenta{the discretization may not necessarily be affine, even though the original continuous-time system is affine}
%The same for the next equation.
%\\
and 
\begin{equation*}
    x_{k+1} = F_h(x_{k},w_k,\extendedcontrol_k)
\end{equation*}
%\red{Q: $n$ instead of $N$? -- corrected}
where $F_h(x_k,w_k,\extendedcontrol_k)$ is \brown{given by} the first $n$ tuples of $\bar{F}_h(\extendedstate,\extendedcontrol)$ corresponding to $x_{k+1}$. \brown The process of obtaining a feedback linearizable discretization for~\eqref{ctssys} can be summarized as a representational commutative diagram shown in Figure~\ref{fig:my_label2}.
\begin{figure}
    \centering
    \begin{tikzpicture}
  \matrix (m) [matrix of math nodes,row sep=3em,column sep=4em,minimum width=2em]
  {
     \dot{x}=f(x)+g(x)u&\\\dot{\extendedstate}= F(\extendedstate)+G(\extendedstate)\extendedcontrol &  \dot\linstate = A\linstate+B\extendedcontrol\\       \extendedstate_{k+1}= \bar F_h(\extendedstate_k,\extendedcontrol_k)& \linstate_{k+1}= A_h\linstate_k+B\extendedcontrol_k \\ x_{k+1} = F_h(x_k,w_k,\extendedcontrol_k)&\\};
  \path[-stealth]
    (m-1-1) edge node [left] {} (m-2-1)
    (m-2-1) edge node [left] {$\discmap$} (m-3-1)
            edge node [above] {$\extdiff,\Psi$} (m-2-2)
    (m-2-1) edge node [below] {CTLS} (m-2-2)
    (m-3-1.east|-m-3-2) edge node [below] {DTLS}
            (m-3-2)
    (m-3-1) edge node [above] {$\extdiff,\Psi$} (m-3-2)
    (m-2-2) edge node [right] {$\discmap_{\extdiff}$} (m-3-2);
    \path[->]
    (m-4-1) edge node [left] {} (m-3-1);    
\end{tikzpicture}
    \caption{Schematic representation of constructing a feedback linearizable discretization scheme for~\eqref{ext_sys}. Both~\eqref{ext_sys} and its discretization~\eqref{retr_desc_lin1} are linearizable by coordinate \brown{change} $\linstate\coloneqq \extdiff(\extendedstate)$ and feedback $\extendedcontrol \coloneqq \extfeed(\brown \extendedstate,v)= \extctra(\extendedstate)+\extctrb(\extendedstate)v$ (CTLS- Continuous-time linear system, DTLS - Discrete-time linear system).} 
    \label{fig:my_label2}
\end{figure}
\section{Illustrating Theorem~\ref{main_result} on Example~\ref{ex_1}}
\label{sec:illustrating_ex}
We apply the results of Theorem~\ref{main_result} on Example~\ref{ex_1} and construct a feedback linearizable discretization scheme \red{for~\eqref{unicycle_extended}.}
%\red{Label problem:it should be 2.7} 

Define $\linstate \coloneqq(z_1,z_2,z_3,z_4,z_5)$ and consider \brown{the following} discretization map on $N$, $$(\linstate,\dot \linstate)\mapsto \discmap_\extdiff(\linstate,\dot \linstate) \coloneqq (\linstate,\linstate+\dot \linstate),$$
then the associated discretization of~\eqref{unicycle_lin} is given by 
\begin{equation}
\label{ex1_dlin}
\begin{split}
\linstate^1_{k+1} &= \linstate^1_{k}+h\linstate^2_{k}\\
\linstate^2_{k+1} &= \linstate^2_{k}+h\linstate^4_{k}\\
\linstate^3_{k+1} &= \linstate^3_{k}+hv^1_{k}\\
\linstate^4_{k+1} &= \linstate^4_{k}+h\linstate^5_{k}\\
\linstate^5_{k+1} &= \linstate^5_{k}+hv^2_{k}
\end{split}    
\end{equation}
where, for each $k\in\N$, $\linstate_k \coloneqq (\linstate^1_k,\linstate^2_k,\linstate^3_k,\linstate^4_k,\linstate^5_k)\approx \linstate(t_k)$ is the approximated state trajectory at $t_k$ and $v_k\coloneqq(v^1_k,v^2_k)$ is the applied control over the interval $[t_{k},t_{k+1}[$, $t_{k+1}= t_k+h$, with $h$ being the sampling period.

%\red{Q: Is it $R_\extdiff$ or $D_\extdiff$ below? -- \magenta{$\discmap_\extdiff$ -- corrected}}

Lifting \red{$\discmap_\extdiff$} onto $M$, define $$(\extendedstate,\dot\extendedstate)\mapsto \discmap(\extendedstate,\dot\extendedstate) \coloneqq \big((\extdiff^{-1}\times\extdiff^{-1})\circ \discmap_\extdiff \circ \T\extdiff\big)(\extendedstate,\dot\extendedstate),$$ where $\extendedstate\coloneqq (x,w)\in\R[4]\times\R$, and $\dot\extendedstate\coloneqq(\dot x,\dot w)\in\R[4]\times\R$.  Then $\discmap$ induces the following discretization scheme on~\eqref{unicycle_extended}:
\begin{equation}
\label{ex1_disc_nonlin}
\extendedstate_{k+1} = \extdiff^{-1}\Big(\extdiff(\extendedstate_k) + h\big(\D\extdiff(\extendedstate_k)\cdot(F(\extendedstate_k)+G(\extendedstate_k)\extendedcontrol_k)\big)\Big).
\end{equation}
For all $k\in\N$, defining $z_{k} \coloneqq \extdiff(\extendedstate_{k})$ and 
using a feedback control 
$\extendedcontrol_k \coloneqq \extctra(\extendedstate_k) + \extctrb(\extendedstate_k)v_k$,
equation~\eqref{ex1_disc_nonlin} can be transformed to \brown{the} linear discrete-time system~\eqref{ex1_dlin}.

% Furthermore, substituting $w_k= u^1_k$ and $u^2_k=v^2_k$, one can show that 
% \begin{equation}
% \lim_{h\lra0} \frac{x_{k+1}-x_{k}}{h} = f(x_k,u_k)    
% \end{equation}
% where $f(x_k,u_k)$ corresponds to the right-hand side of~\eqref{ex1_dyn_cts} and thus~\eqref{ex1_disc_nonlin} results in a feedback linearizable discretization for~\eqref{ex1_dyn_cts} that is accurate up to the first order.
% The above example shows that one can lift the discretization scheme to construct feedback linearizable discretization maps for the extended dynamical system. 
\section{Simulation Results}
We demonstrate the discretizing scheme by implementing \brown{it} on a stabilizing problem. The initial condition was chosen as $\extendedstate(0) = (x(0),w(0)) = ((0.5,0.2,0.1,0.2),0)$. The scheme was simulated over an interval of $10$ seconds with a stepsize of $h=10^{-2}$ seconds. The stabilizing control law was chosen as 
$$\pmat{\extendedcontrol_k^1\\\extendedcontrol_k^2}\coloneqq \pmat{1&x_k^4\\ w_k& 1+x_k^3}^{-1}\pmat{v_k^1-(1+x_k^3)(w_k)^2\\v_k^2},$$
with $v_k \coloneqq (v_k^1,v_k^2) = (-(10z_k^1+10z_k^2+10z_k^3),-(10z_k^4+10z_k^5))$ and \red{$z_k = \extdiff(\extendedstate_k)$}. 

%\red{Above $z_k = \extdiff(x_k, w_k)$ or  $z_k = \extdiff(\xi_k)$}

The state and control plots for the discretized system are shown in Figures~\ref{fig:state_ex1} and~\ref{fig:control_ex1} respectively. The global error $\norm{\extendedstate(t_k)-\extendedstate_k}$, where $\extendedstate(t_k)$ is the exact trajectory 
(obtained by ODE45 solver of MATLAB) of~\eqref{unicycle_extended} sampled at $t_k$, is plotted in Figure~\ref{fig:error_ex1}. 
\begin{figure}
    \centering
   \includegraphics[width=\linewidth]{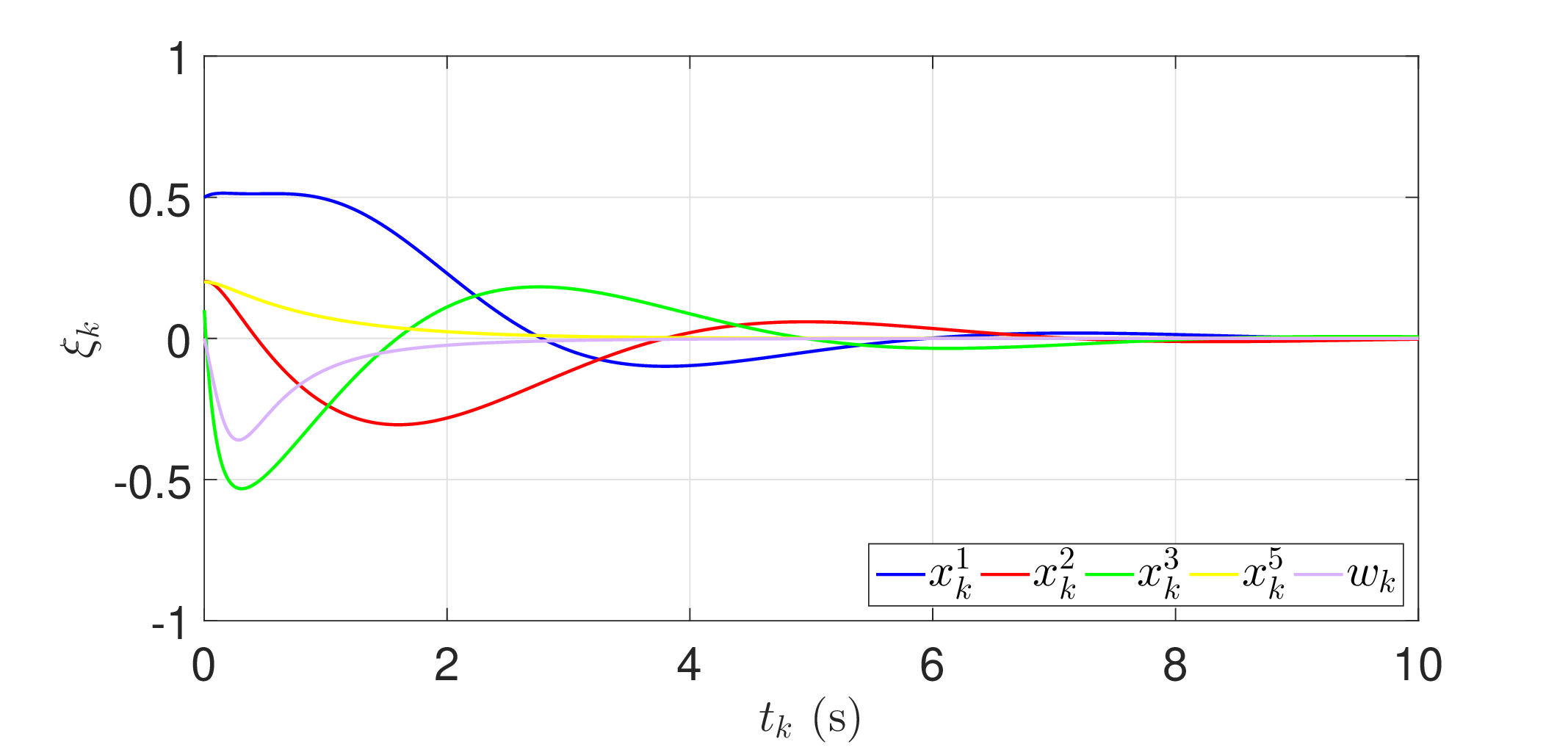}
    \caption{System states $\xi_k\coloneqq (x_k,w_k)$ for~\eqref{ex1_disc_nonlin} for a stepsize $h=10^{-2}$ and $t_k\in[0,10]$.}
    \label{fig:state_ex1}
\end{figure}
\begin{figure}
    \centering
   \includegraphics[width=\linewidth]{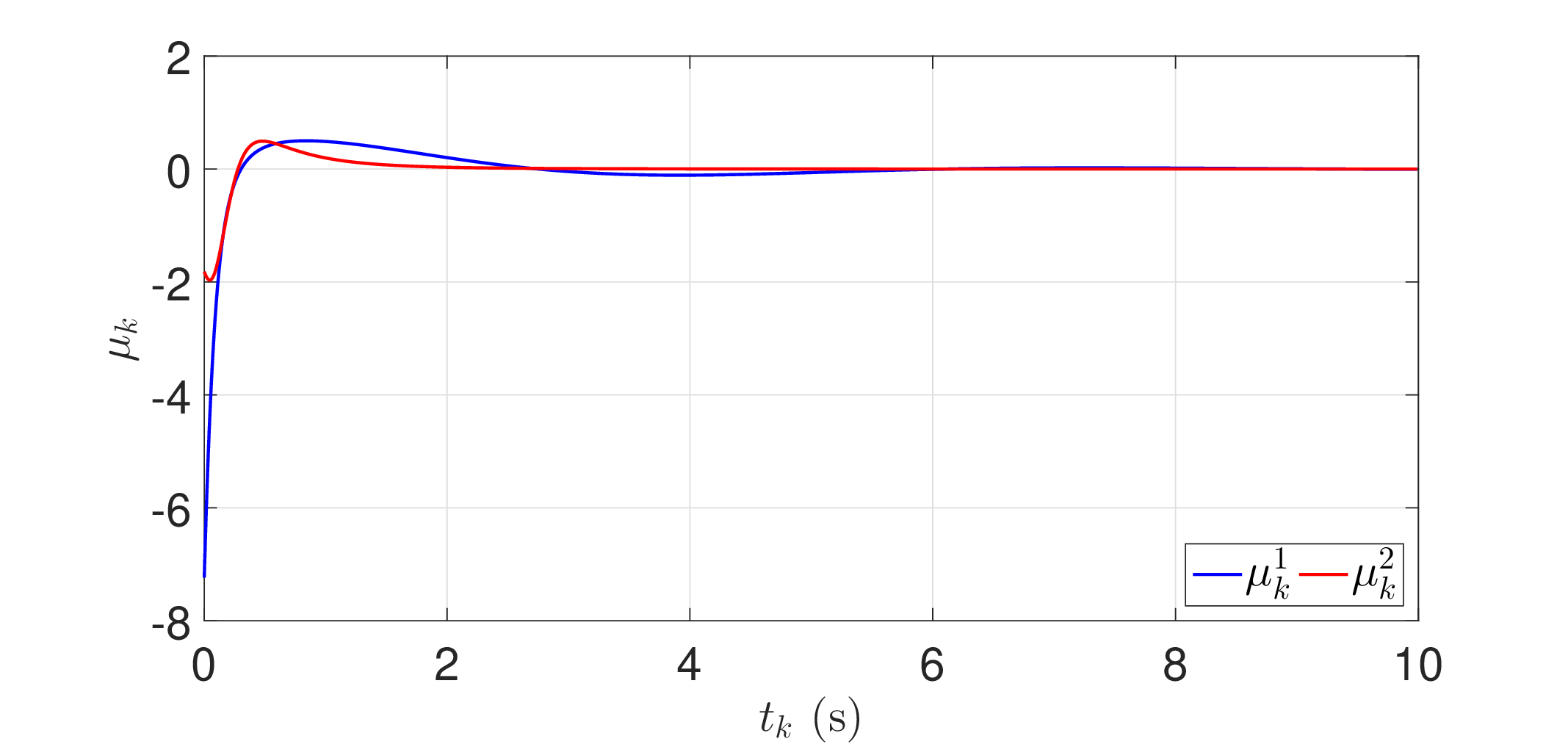}
    \caption{Control input $\mu_k\coloneqq (\mu_k^1,\mu_k^2)$ for~\eqref{ex1_disc_nonlin} for a stepsize $h=10^{-2}$ and $t_k\in[0,10]$.}
    \label{fig:control_ex1}
\end{figure}
\begin{figure}
    \centering   \includegraphics[width=\linewidth]{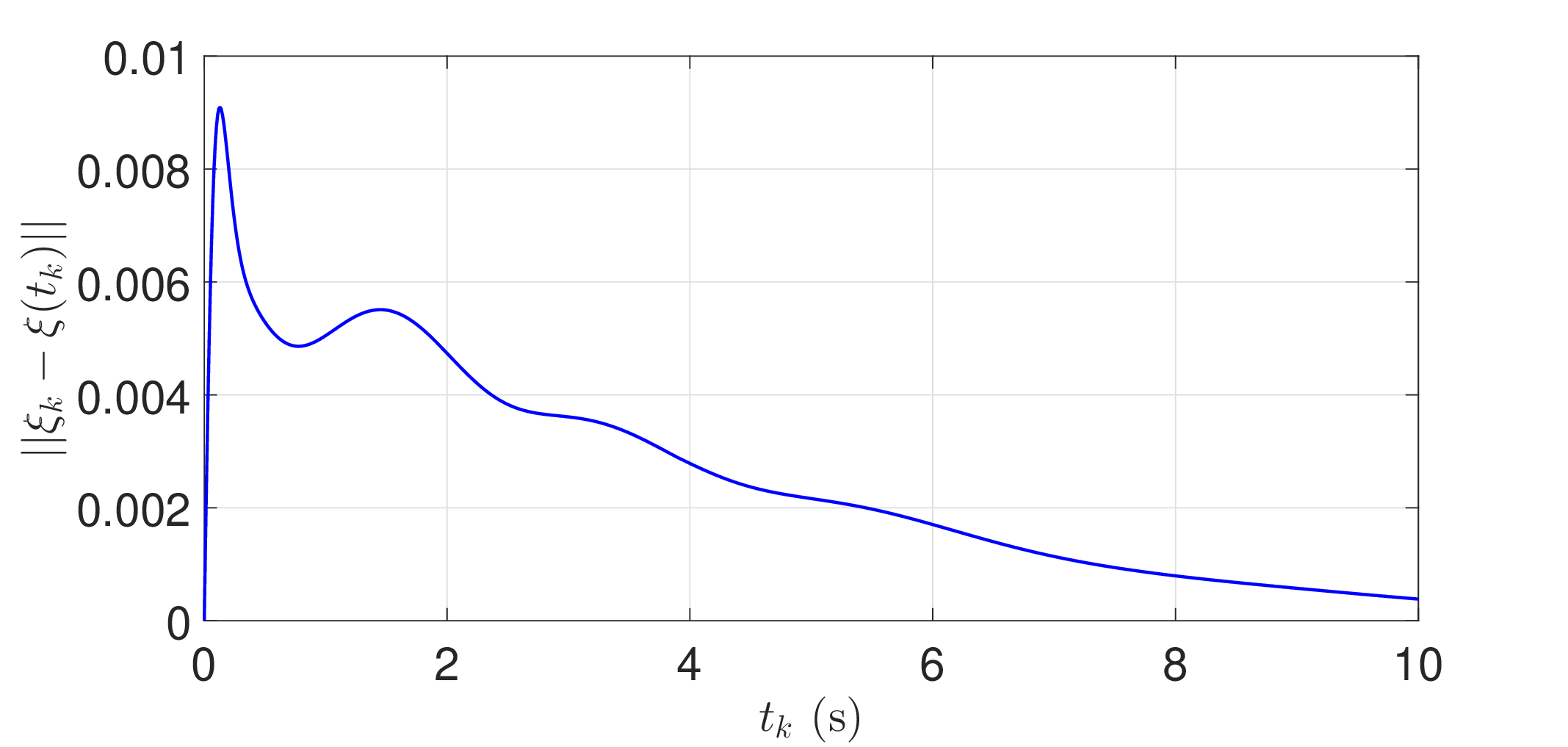}
    \caption{Global Error  $\norm{\extendedstate(t_k)-\extendedstate_k}$ for~\eqref{ex1_disc_nonlin}, for a stepsize $h=10^{-2}$ and $t_k\in[0,10]$.}
    \label{fig:error_ex1}
\end{figure}
\section{Conclusion}
In this article, we have extended the results of \cite{retr_disc_map} to the dynamical feedback linearizable systems. Theorem~\ref{main_result} allows us to construct discretization schemes that are feedback linearizable for a class of nonlinear systems linearizable by dynamic feedback. One of the key features of the results presented here is that we do not assume any invertibility \brown{property} on the type of feedback considered. This allows the result to apply to a class of systems including both endogenous and exogenous feedback. Although the results are presented here for the control-affine form, they hold for general nonlinear systems. To illustrate our results, we implement this on a stabilization problem of a dynamical feedback linearizable system. The simulation was run for 10 seconds and the trajectories and error magnitudes were plotted. From the error plot one can see that for a stepsize $h=10^{-2}$ seconds the error is fairly of the order of $10^{-2}$. As a future work similar to \cite{retr_disc_map}, one can construct higher order discretization by using multistep discretization while preserving feedback linearizability.
\section*{Appendix}
Consider the Euler Discretization~\eqref{uni_dis_eul} for system~\eqref{unicycle_extended}. Denoting it compactly, the discrete-time system is given by 
\begin{equation}
    \extendedstate_{k+1} = F_{h}(\extendedstate_k,\extendedcontrol_k).
\end{equation}
One now utilizes the necessary and sufficient conditions from \cite{grizzle1986feedback} to prove the feedback linearizability of~\eqref{uni_dis_eul}. In this direction, the Jacobian of $F_h$ is given by $\D F_h(x,w,\extendedcontrol) = \pmat{\pdf{F_h}{x}(x,w,\extendedcontrol) &\pdf{F_h}{u}(x,w,\extendedcontrol)}$ with
\begin{multline*}
    \pdf{F_h}{x}(x,w,\extendedcontrol) =\\ \pmat{1&h(1+2(x_3+x_4w))&2hx_2&2hx_2w&2hx_2x_4\\
          0&1&h&w&x_4\\
     0&0&1&0&0\\
     0&0&hw&1&h(1+x_3)\\
     0&0&0&0&1}
\end{multline*}
\begin{equation*}
   \pdf{F_h}{u}(x,w,\extendedcontrol) = 
   \pmat{0&0&h&0&0\\
         0&0&0&0&h}^\top 
\end{equation*}
and its Kernel distribution is given by 
\begin{equation*}
    \K = \colspan
\left(\pmat{\star_{1}&\star_{2}&0&h^2(1+x_3)&-h&0&1\\\star_3&\star_4&-h&h^2w&0&1&0}^\top\right)
\end{equation*}
where $\star_1 \coloneqq -h(1+2(x_3+x_4w))(h^2x_4-h^3w(1+x_3))-2h^3x_2w(1+x_3)+2h^2x_2x_4$, $\star_2 \coloneqq -h^3w(1+x_3)+h^2x_4$, $\star_3 \coloneqq -h(1+2(x_3+x_4w))(h^2-h^3w)+2x_2(h^2-h^3w)$ and $\star_4\coloneqq h^2-h^3w $. Initiating a sequence of distribution as given in Theorem~3.1 form \cite{grizzle1986feedback} we have 
\begin{equation*}
    \dis{0} = \colspan\left(\pmat{0&0&0&0&0&1&0\\0&0&0&0&0&0&1}^\top\right)
\end{equation*}
Since $\dis{0}+\K$ is involutive and $\dis{0}\cap\K=\{0\}$, being the zero distribution, is constant dimensional, we have 
\begin{equation*}
    \dis{1} = \colspan\left(\pmat{0&0&0&0&0&1&0\\0&0&0&0&0&0&1\\
    0&0&0&0&h&0&0\\0&0&h&0&0&0&0}^\top\right)
\end{equation*}
involutive distribution. However, 
\begin{multline*}
    \dis{1}+\K =\\ \colspan\left(\pmat{0&0&0&0&0&1&0\\0&0&0&0&0&0&1\\
    0&0&0&0&h&0&0\\0&0&h&0&0&0&0\\\star_{1}&\star_{2}&0&h^2(1+x_3)&-h&0&1\\\star_3&\star_4&-h&h^2w&0&1&0}^\top\right)
\end{multline*}
is not involutive, therefore using Theorem~3.1 from \cite{grizzle1986feedback}, \brown{we conclude that}~\eqref{uni_dis_eul} is not feedback linearizable.

%\red{Check references:  21 and 28 is the same reference} 
\bibliographystyle{IEEEtran}
\bibliography{IEEEabrv,name1}
\end{document}